\documentclass[11pt]{article}
\usepackage{geometry}
\usepackage{graphicx}
\usepackage{color}
\usepackage{amsthm}

\newtheoremstyle{willthm}%      name
     {4pt}%       Space above
     {4pt}%       Space below     
     {}%          Body font
     {}%          Indent amount (empty = no indent, \parindent = para indent)
     {\bfseries}% Thm head font
     {.}%         Punctuation after thm head
     { }%	  Space after thm head: " " = normal interword space; %       \newline = linebreak
     {}%          Thm head spec (can be left empty, meaning `normal')

\theoremstyle{willthm}

\newtheorem{thm}{Theorem}[section]
\newtheorem{lem}[thm]{Lemma}
\newtheorem{cor}[thm]{Corollary}

\newtheorem{rem}[thm]{Remark}

\newtheorem{open}[thm]{Open Problem}

\usepackage{amsmath}
\usepackage{amsthm}
\usepackage{amssymb}
\usepackage{marginnote}
\newcommand{\dft}[1]{\textbf{\textit{#1}}}

\makeatletter
\let\Pr\@undefined
\makeatother

\DeclareMathOperator{\Pr}{\mathbf{P}}

\newcommand{\abs}[1]{\left|#1\right|}

\newcommand{\floor}[1]{\left\lfloor#1\right\rfloor}

\newcommand{\set}[1]{\left\{#1\right\}}
\newcommand{\st}{\,\middle|\,}

\newcommand{\paren}[1]{\left(#1\right)}
\newcommand{\too}{\,\longrightarrow\,}

\newcommand{\e}{\varepsilon}
\renewcommand{\bar}[1]{\overline{#1}}

\usepackage{algorithmic}
\usepackage{algorithm}

\newcommand{\AMSM}{\textbf{AMSM}}
\newcommand{\ASR}{\textbf{ASA}}

\newcommand{\GTGMSM}[1]{Gap$_{#1}$-$3$G-MSM}
\newcommand{\GTGMSS}[1]{Gap$_{#1}$-$3$G-MSS}
\newcommand{\GTDMT}[1]{Gap$_{#1}$-$3$DM-$3$}
\newcommand{\TGSM}{$3$GSM}
\newcommand{\TPSA}{$3$PSA}
\newcommand{\TGMSM}{$3$G-MSM}
\newcommand{\TGMSS}{$3$G-MSS}
\newcommand{\TPSAMSM}{\TPSA-MSM}
\newcommand{\TPSAMSS}{\TPSA-MSS}

\DeclareMathOperator{\ins}{ins}
\DeclareMathOperator{\MSM}{MSM}
\DeclareMathOperator{\MSS}{MSS}
\DeclareMathOperator{\opt}{opt}

\DeclareMathOperator{\sOh}{\tilde{\mathcal{O}}}
\DeclareMathOperator{\stab}{stab}

\title{It's Not Easy Being Three: The Approximability of Three-Dimensional Stable Matching Problems}
\author{Rafail Ostrovsky \thanks{University of California, Los Angeles (Departments of Computer Science and Mathematics). Work supported in part by NSF grants 09165174, 1065276, 1118126 and 1136174, US-Israel BSF grant 2008411, OKAWA Foundation Research Award, IBM Faculty Research Award, Xerox Faculty Research Award, B. John Garrick Foundation Award, Teradata Research Award, and Lockheed-Martin Corporation Research Award. This material is based upon work supported by the Defense Advanced Research Projects Agency through the U.S. Office of Naval Research under Contract N00014-11-1-0392. The views expressed are those of the author and do not reflect the official policy or position of the Department of Defense or the U.S. Government.}
\and
Will Rosenbaum \thanks{University of California, Los Angeles (Department of Mathematics).}
}
\date{\today}

\begin{document}
  \maketitle
  \begin{abstract}
    In 1976, Knuth \cite{Knuth97} asked if the stable marriage problem (SMP) can be generalized to marriages consisting of 3 genders. In 1988, Alkan \cite{Alkan88} showed that the natural generalization of SMP to 3 genders (\TGSM) need not admit a stable marriage. Three years later, Ng and Hirschberg \cite{NH91} proved that it is NP-complete to determine if given preferences admit a stable marriage. They further prove an analogous result for the 3 person stable assignment (\TPSA) problem.

    In light of Ng and Hirschberg's NP-hardness result for \TGSM\ and \TPSA, we initiate the study of approximate versions of these problems. In particular, we describe two optimization variants of \TGSM\ and \TPSA: \emph{maximally stable marriage/matching} (\emph{MSM}) and \emph{maximum stable submarriage/submatching} (\emph{MSS}). We show that both variants are NP-hard to approximate within some fixed constant factor. Conversely, we describe a simple polynomial time algorithm which computes constant factor approximations for the maximally stable marriage and matching problems. Thus both variants of MSM are APX-complete.
  \end{abstract}

  %%%%%%%%%%%%%%%%%%%%%%%%%%%%%%%%%%%%%%%%%%%%%%%%%%%%%%%%%%%%%%%%%%%%%
  %% Introduction
  %%%%%%%%%%%%%%%%%%%%%%%%%%%%%%%%%%%%%%%%%%%%%%%%%%%%%%%%%%%%%%%%%%%%%

  \section{Introduction}
  \label{sec:introduction}

  %% Previous Work
  \subsection{Previous Work}
  \label{sec:previous-work}

  Since Gale and Shapley first formalized and studied the stable marriage problem (SMP) in 1962 \cite{GS62}, many variants of the SMP have emerged (see, for example,  \cite{GI89, Knuth97, Manlove13, RS92}). While many of these variants admit efficient algorithms, two notably do not\footnote{Assuming, of course, P$\neq$NP!}: (1) incomplete preferences with ties \cite{IMMM99}, and (2) $3$ gender stable marriages (\TGSM) \cite{NH91}. 

  In the case of incomplete preferences with ties, it is NP-hard to find a maximum cardinality stable marriage \cite{IMMM99}. The intractability of exact computation for this problem led to the study of approximate versions of the problem. These investigations have resulted in hardness of approximation results \cite{IMO09, Yanagisawa07} as well as constant factor approximation algorithms \cite{Kiraly11, Kiraly13, Paluch14, Yanagisawa07}. 

  In \TGSM, players are one of three genders: women, men, and dogs (as suggested by Knuth). Each player holds preferences over the set of \emph{pairs} of players of the other two genders. The goal is to partition the players into families, each consisting of one man, one woman, and one dog, such that no triple mutually prefer one another to their assigned families. In 1988, Alkan showed that for this natural generalization of SMP to three genders, there exist preferences which do not admit a stable marriage \cite{Alkan88}. In 1991, Ng and Hirschberg showed that, in fact, it is NP-complete to determine if given preferences admit a stable marriage \cite{NH91}. They further generalize this result to the three person stable assignment problem (\TPSA).  In \TPSA, each player ranks all pairs of other players without regard to gender. The goal is to partition players into disjoint triples where again, no three players mutually prefer each other to their assigned triples. 

  Despite the advances for stable marriages with incomplete preferences and ties (see \cite{Manlove13} for an overview of relevant work), analogous approximability results have not been obtained for $3$ gender variants of the stable marriage problem. In this paper, we achieve the first substantial progress towards understanding the approximability of \TGSM\ and \TPSA. 

  %% Overview of Our Results
  \subsection{Overview of our results}

  \subsubsection{$3$ gender stable marriages (\TGSM)}
  \label{sec:\TGSM-overview}

  We formalize two optimization variants of \TGSM: maximally stable marriage (\TGMSM) and maximum stable submarriage (\TGMSS). For \TGMSM, we seek a perfect ($3$ dimensional) marriage which minimizes the number of unstable triples---triples of players who mutually prefer each other to their assigned families. For \TGMSS, we seek a largest cardinality \emph{sub}marriage which contains no unstable triples among the married players. Exact computation of both of these problems is NP-hard by Ng and Hirschberg's result \cite{NH91}. Indeed, exact computation of either allows one to detect the existence of a stable marriage. 

  We obtain the following inapproximability result for \TGMSM\ and \TGMSS.

  \begin{thm}[Special case of Theorem \ref{thm:main-3gsm-lb}]
    \label{thm:msm-lb-informal}
    There exists an absolute constant $c < 1$ such that it is NP-hard to approximate \TGMSM\ and \TGMSS\ to within a factor $c$.
  \end{thm}

  In fact, we prove a slightly stronger result for \TGMSM\ and \TGMSS. We show that the problem of determining if given preferences admit a stable marriage or if all marriages are ``far from stable'' is NP-hard. See Section \ref{sec:\TGSM} and Theorem \ref{thm:main-3gsm-lb} for the precise statements. In the other direction, we describe a polynomial time constant factor approximation algorithm for \TGMSM.

  \begin{thm}
    \label{thm:amsm-informal}
    There exists a polynomial time algorithm, \AMSM, which computes a $\frac 4 9$-factor approximation to \TGMSM.
  \end{thm}

  \begin{cor}
    \TGMSM\ is APX-complete.
  \end{cor}

  \subsubsection{Three person stable assignment (\TPSA)}
  \label{sec:3sr-overview}

  We also consider the three person stable assignment problem (\TPSA). In this problem, players rank all pairs of other players and seek a ($3$ dimensional) matching---a partition of players into disjoint triples. Notions of stability, maximally stable matching, and maximum stable submatching are defined exactly as the analogous notions for \TGSM. We show that Theorems \ref{thm:msm-lb-informal} and \ref{thm:amsm-informal} have analogues with \TPSA:

  \begin{thm}
    \label{thm:sr-msm-lb-informal}
    There exists a constant $c < 1$ such that it is NP-hard to approximate \TPSAMSM\ and \TPSAMSS\ to within a factor $c$.
  \end{thm}

  \begin{thm}
    \label{thm:asr-informal}
    There exists a polynomial time algorithm, \ASR, which computes a $\frac 4 9$-factor approximation to \TPSAMSM.
  \end{thm}

  Our proofs of the lower bounds in Theorems \ref{thm:msm-lb-informal} and  \ref{thm:sr-msm-lb-informal} use a reduction from the $3$ dimensional matching problem ($3$DM) to \TGMSM. Kann \cite{Kann91} showed that Max-$3$DM is Max-SNP complete. Thus, by the PCP theorem \cite{ALMSS98, AS98} and \cite{Dinur07}, it is NP-complete to approximate Max-$3$DM to within some fixed constant factor. Our hardness of approximation results then follow from a reduction from $3$DM to \TGMSM.

  Theorems \ref{thm:amsm-informal} and \ref{thm:asr-informal} follow from a simple greedy algorithm. Our algorithm constructs marriages (or matchings) by greedily finding triples whose members are guaranteed to participate in relatively few unstable triples. Thus, we are able to efficiently construct marriages (or matchings) with a relatively small fraction of blocking triples.

  %%%%%%%%%%%%%%%%%%%%%%%%%%%%%%%%%%%%%%%%%%%%%%%%%%%%%%%%%%%%%%%%%%%%%
  %% Background and Definitions
  %%%%%%%%%%%%%%%%%%%%%%%%%%%%%%%%%%%%%%%%%%%%%%%%%%%%%%%%%%%%%%%%%%%%%

  \section{Background and Definitions}
  \label{sec:background}

  \subsection{3 Gender Stable Marriage (\TGSM)}
  \label{sec:\TGSM}

  In the 3 gender stable marriage problem, there are disjoint sets of \dft{women}, \dft{men}, and \dft{dogs} denoted by $A$ (for Alice), $B$ (for Bob), and $D$ (for Dog), respectively. We assume $\abs{A} = \abs{B} = \abs{D} = n$, and we denote the collection of \dft{players} by $V = A \cup B \cup D$. A \dft{family} is a triple $abd$ consisting of one woman $a \in A$, one man $b \in B$, and one dog $d \in D$. A \dft{submarriage} $S$ is a set of pairwise disjoint families. A \dft{marriage} $M$ is a maximal submarriage---that is, one in which every player $v \in V$ is contained in some (unique) family so that $\abs{M} = n$. Given a submarriage $S$, we denote the function $p_S : V \to V^2 \cup \set{\varnothing}$ which assigns each player $v \in V$ to their partners in $S$, with $p_S(v) = \varnothing$ if $v$ is not contained in any family in $S$.

  Each player $v \in V$ has a \dft{preference}, denoted $\succ_v$ over pairs of members of the other two genders. That is, each woman $a \in A$ holds a total order $\succ_a$ over $B \times D \cup \set{\varnothing}$, and similarly for men and dogs. We assume that each player prefers being in some family to having no family. For example, $bd \succ_a \varnothing$ for all $a \in A$, $b \in B$ and $d \in D$. An instance of the \dft{three gender stable marriage problem} (\dft{\TGSM}) consists of $A$, $B$, and $D$ together with a set $P = \set{\succ_v \st v \in V}$ of preferences for each $v \in V$. 

  Given a submarriage $S$, a triple $abd$ is an \dft{unstable triple} if $a$, $b$ and $d$ each prefer the triple $abd$ to their assigned families in $S$. That is, $abd$ is unstable if and only if $bd \succ_a p_S(a)$, $ad \succ_b p_S(b)$, and $ab \succ_d p_S(d)$. A triple $abd$ which is not unstable is \dft{stable}. In particular, $abd$ is stable if at least one of $a$, $b$ and $d$ prefers their family in $S$ to $abd$. Let $A_S$, $B_S$ and $D_S$ be the sets of women, men and dogs (respectively) which have families in $S$. A submarriage $S$ is \dft{stable} if there are no unstable triples in $A_S \times B_S \times D_S$. 

  Unlike the two gender stable marriage problem, this three gender variant arbitrary preferences need not admit a stable marriage. In fact, for some preferences, every marriage has many unstable triples (see Section \ref{sec:many-blockers}). Thus we consider two optimization variants of the three gender stable marriage problem.

  \subsubsection{Maximally Stable Marriage (\TGMSM)} The \dft{maximally stable marriage problem} (\TGMSM) is to find a marriage $M$ with the maximum number of stable triples with respect to given preferences $P$. For fixed preferences $P$ and marriage $M$, the \dft{stability} of $M$ with respect to $P$ is the number of stable triples in $A \times B \times D$:
  \[
  \stab(M) = \abs{\set{abd \st abd \text{ is stable}}}.
  \]
  Thus, $M$ is stable if and only if $\stab(M) = n^3$. Dually, we define the \dft{instability} of $M$ by $\ins(M) = n^3 - \stab(M)$. For fixed preferences $P$, we define
  \[
  \MSM(P) = \max \set{\stab(M) \st M \text{ is a marriage}}.
  \]
  For preferences $P$ and fixed $c < 1$, we define \dft{\GTGMSM{c}} to be the problem of determining if $\MSM(P) = n^3$ or $\MSM(P) \leq c n^3$. 

  \subsubsection{Maximum Stable Submarriage} The \dft{maximum stable submarriage problem} (\TGMSS) is to find a maximum cardinality stable submarriage $S$. We denote
  \[
  \MSS(P) = \max \set{\abs{S} \st S \text{ is a stable submarriage}}
  \]
  Note that $P$ admits a stable marriage if and only if $\MSS(P) = n$. For fixed $c < 1$, we define \dft{\GTGMSS{c}} to be the problem of determining if $\MSS(P) = n$ or if $\MSS(P) \leq c n$.

  \subsection{Three person stable assignment (\TPSA)}
  \label{sec:\TPSA}

  In the \dft{three person stable assignment problem} (\dft{\TPSA}), there is a set $U$ of $\abs{U} = 3n$ players who wish to be partitioned into $n$ disjoint triples. For a set $C \subseteq U$, we denote the set of $k$-subsets of $C$ by $\binom{C}{k}$. A \dft{submatching} is a set $S \subseteq \binom{U}{3}$ of disjoint triples in $U$. A \dft{matching} $M$ is a maximal submatching---a submatching with $\abs{M} = n$. Given a submatching $S$, $U_S$ is the set of players contained in some triple in $S$:
  \[
  U_S = \set{u \in U \st u \in t \text{ for some } t \in S}.
  \]
  Each player $u \in U$ holds preferences among all pairs of potential partners. That is, each $u \in U$ holds a linear order $\succ_u$ on $\binom{U\setminus\set{u}}{2} \cup \set{\varnothing}$. We assume that each player prefers every pair to an empty assignment. Given a set $P$ of preferences for all the players and a submatching $S$, we call a triple $uvw \in \binom{U_S}{3}$ \dft{unstable} if each of $u$, $v$ and $w$ prefer the triple $uvw$ to their assigned triples in $S$. Otherwise, we call $uvw$ \dft{stable}. A submatching $S$ is \dft{stable} if it contains no unstable triples in $\binom{U_S}{3}$. We define the \dft{stability} of $S$ by
  \[
  \stab(S) = \abs{\set{uvw \in \binom{U_S}{3} \st uvw \text{ is stable}}}.
  \]
  Dually, the \dft{instability} of $S$ is $\ins(S) = \binom{\abs{S}}{3} - \stab(S)$. 

  The \dft{maximally stable matching problem} (\dft{\TPSA-MSM}) is to find a matching $M$ which maximizes $\stab(M)$. The \dft{maximum stable submatching problem} (\dft{\TPSA-MSS}) is to find a stable submatching $S$ of maximum cardinality.

  \begin{rem}
    \label{rem:3sr}
    We may consider a variant of \TPSA\ with \dft{unacceptable partners}. In this variant, each player $u \in U$ ranks only a subset of $\binom{U \setminus\set{u}}{2}$, and prefers being unmatched to unranked pairs. \TGSM\ is a special case of this variant where $U = A \cup B \cup D$ and each player ranks precisely those pairs consisting of one player of each other gender. This observation will make our hardness results for \TGSM\ easily generalize to \TPSA.
  \end{rem}

  \subsection{Hardness of \GTDMT{c}}
  \label{sec:hoa}

  Our proofs of Theorems \ref{thm:msm-lb-informal} and \ref{thm:sr-msm-lb-informal} use a reduction from the three dimensional matching problem ($3$DM). In this section, we briefly review $3$DM, and state the approximability result we require for our lower bound results. 

  Let $W$, $X$ and $Y$ be finite disjoint sets with $\abs{W} = \abs{X} = \abs{Y} = m$. Let $E \subseteq W \times X \times Y$ be a set of edges. A \dft{matching} $M \subseteq E$ is a set of disjoint edges. The \dft{maximum $3$ dimensional matching problem} (\dft{Max-$3$DM}) is to find (the size of) a matching $M$ of largest cardinality in $E$. \dft{Max-$3$DM-$3$} is the restriction of Max-$3$DM to instances where each element in $W \cup X \cup Y$ is contained in at most $3$ edges. For a fixed constant $c < 1$, we define \dft{\GTDMT{c}} to be the problem of determining if an instance $I$ of Max-$3$DM-$3$ has a perfect matching (a matching $M$ of size $m$) or if every matching has size at most $c m$. 

  \begin{thm}
    \label{thm:gap-3dm-3-hardness}
    There exists an absolute constant $c < 1$ such that \GTDMT{c}\ is NP-hard.
  \end{thm}

  Kann showed that Max-$3$DM-$3$ is Max-SNP complete\footnote{The complexity class Max-SNP was introduced by Papadimitriou and Yannakakis in \cite{PY91}, where the authors also show that Max-$3$SAT-$B$ is Max-SNP complete.} by giving an $L$-reduction from Max-$3$SAT-$B$ to Max-$3$DM-$3$ \cite{Kann91}. By the celebrated PCP theorem \cite{ALMSS98, AS98} and \cite{Dinur07}, Kann's result immediately implies that Max-$3$DM-$3$ is NP-hard to approximate to within some fixed constant factor. However, Kann's reduction gives a slightly weaker result than Theorem \ref{thm:gap-3dm-3-hardness}. In Kann's reduction, satisfiable instances of $3$SAT-$B$ do not necessarily reduce to instances of $3$DM-$3$ which admit perfect matchings. In Appendix \ref{sec:gap-hardness}, we describe how to alter Kann's reduction so that satisfiable instances of $3$SAT-$B$ admit perfect matchings, while far-from-satisfiable instances are far from admitting perfect matchings.

  %%%%%%%%%%%%%%%%%%%%%%%%%%%%%%%%%%%%%%%%%%%%%%%%%%%%%%%%%%%%%%%%%%%%%
  %% Hardness of Approximation
  %%%%%%%%%%%%%%%%%%%%%%%%%%%%%%%%%%%%%%%%%%%%%%%%%%%%%%%%%%%%%%%%%%%%%

  \section{Hardness of Approximation}
  \label{sec:hardness}

  In this section, we prove the main hardness of approximation results. Specifically, we will prove the following theorems.

  \begin{thm}
    \label{thm:main-3gsm-lb}
    \GTGMSM{c}\ and \GTGMSS{c}\ are NP-hard.
  \end{thm}

  \begin{thm}
    \label{thm:main-3sr-lb}
    There exists an absolute constant $c < 1$ such that Gap$_c$-\TPSA-MSM and Gap$_c$-\TPSA-MSS are NP-hard.
  \end{thm}

  %% Preferences with Many Unstable Triples

  \subsection{Preferences for \TGSM\ with Many Unstable Triples}
  \label{sec:many-blockers}

  \begin{thm}
    \label{thm:msm}
    There exist preferences $P$ for \TGSM\ and a constant $c < 1$ for which $\MSM(P) \leq c n^3$.
  \end{thm}

  We describe preferences $P$ for which every marriage has $\Omega(n^3)$ blocking triples below. Assuming $n$ is even, we partition each gender into two equal sized sets $A = A_1 \cup A_2$, $B = B_1 \cup B_2$, and $D = D_1 \cup D_2$
    \begin{center}
      \begin{tabular}{l|rrr}
      player & \multicolumn{3}{c}{preferences}\\
      \hline
      $a_1 \in A_1$ & $B_1 D_1$ & $B_2 D_2$ & $\cdots$\\
      $a_2 \in A_2$ & $B_2 D_1$ & $\cdots$\\
      $b_1 \in B_1$ & $A_1 D_1$ & $\cdots$\\
      $b_2 \in B_2$ & $A_1 D_2$ & $A_2 D_1$ & $\cdots$\\
      $d_1 \in D_1$ & $A_2 B_2$ & $A_1 D_1$ & $\cdots$\\
      $d_2 \in D_2$ & $A_1 B_2$ & $\cdots$
    \end{tabular}
    \end{center}
    The sets appearing in the preferences indicate that the player prefers all pairs in that set (in any order) followed by the remaining preferences. For example, all $a_1 \in A_1$ prefer all partners $bd \in B_1 \times D_1$, followed by all partners in $B_2 \times D_2$, followed by all other pairs in arbitrary order. Within $B_1 \times D_1$ and $B_2 \times D_2$, $a_1$'s preferences are arbitrary. The full proof of Theorem \ref{thm:msm} is given in Appendix \ref{sec:unstable}.

  %% The Embedding

  \subsection{The Embedding}
  \label{sec:embedding}

  We now describe an embedding of $3$DM-$3$ into \TGMSM. Our embedding is a modification of the embedding described by Ng and Hirschberg \cite{NH91}. Let $I$ be an instance of $3$DM-$3$ with ground sets  $W, X, Y$ and edge set $E$. We assume $\abs{W} = \abs{X} = \abs{Y} = m$. We will construct an instance $f(I)$ of \TGMSM\ with sets $A, B$ and $D$ of women, men and dogs of size $n = 6 m$ and suitable preferences $P$. We divide each gender into two sets $A = A^1 \cup A^2$, $B = B^1 \cup B^2$ and $D = D^1 \cup D^2$ where $\abs{A^j} = \abs{B^j} = \abs{D^j} = 3 m$ for $j = 1, 2$. Let $W = \set{a_1, a_2, \ldots, a_m}$, $X = \set{b_1, b_2, \ldots, b_m}$ and $Y = \set{d_1, d_2, \ldots, d_m}$, and denote 
  \[
  E = \bigcup_{i = 1}^n \set{a_i b_{k_1} d_{\ell_1},\ a_i b_{k_2} d_{\ell_2},\ a_i b_{k_3} d_{\ell_3}}.
  \]
  Without loss of generality, we assume each $a_i$ is contained in exactly $3$ edges by possibly increasing the multiplicity of edges containing $a_i$. For $j = 1, 2$, we form sets
  \begin{align*}
    A^j = \set{a_i^j[k] \st i \in [n],\ k \in [3]}, \quad B^j = \set{b_i^j,\ w_i^j,\ y_i^j \st i \in [n]}, \quad D^j &= \set{d_i^j,\ x_i^j,\ z_i^j \st i \in [n]}
  \end{align*}
  for $j = 1, 2$. We now define preferences for each set of players, beginning with those in $A$.
  \[
  \begin{array}{l|cccccc}
    a_i^1[m] & w_i^1 x_i^1 & y_i^1 z_i^1 & b_{k_m}^1 d_{\ell_m}^1 & B^1 D^1 & B^2 D^2 & \cdots\\[5pt]
    a_i^2[m] & w_i^2 x_i^2 & y_i^2 z_i^2 & b_{k_m}^2 d_{\ell_m}^2 & B^2 D^1 & \cdots & 
  \end{array}
  \]
  The players in $B$ have preferences given by
  \[
  \begin{array}{l|cccccc}
    w_i^1 & a_i^1[1] x_i^1 & a_i^1[2] x_i^1 & a_i^1[3] x_i^1 & A^1 D^1 & \cdots\\[5pt]
    y_i^1 & a_i^1[1] z_i^1 & a_i^1[2] z_i^1 & a_i^1[3] z_i^1 & A^1 D^1 & \cdots\\[5pt]
    b_i^1 & A^1 D^1 & \cdots\\[5pt]
    w_i^2 & a_i^2[1] x_i^2 & a_i^2[2] x_i^2 & a_i^2[3] x_i^2 & A^1 D^2 & A^2 D^1 & \cdots\\[5pt]
    y_i^2 & a_i^2[1] z_i^2 & a_i^2[2] z_i^2 & a_i^2[3] z_i^2 & A^1 D^2 & A^2 D^1 & \cdots\\[5pt]
    b_i^2 & A^1 D^2 & A^2 D^1 & \cdots
  \end{array}
  \]
  The preferences for $D$ are given by
  \[
  \begin{array}{l|cccccc}
    x_i^1 & a_i^1[3] w_i^1 & a_i^1[2] w_i^1 & a_i^1[3] w_i^1 & A^2 B^2 & A^1 B^1 &\cdots\\[5pt] 
    z_i^1 & a_i^1[3] y_i^1 & a_i^1[2] y_i^1 & a_i^1[3] y_i^1 & A^2 B^2 & A^1 B^1 & \cdots \\[5pt] 
    d_i^1 & A^2 B^2 & A^1 B^1 & \cdots \\[5pt]
    x_i^2 & a_i^2[3] w_i^2 & a_i^2[2] w_i^2 & a_i^2[3] w_i^2 & A^1 B^2 & \cdots \\[5pt] 
    z_i^2 & a_i^2[3] y_i^2 & a_i^2[2] y_i^2 & a_i^2[3] y_i^2 & A^1 B^2 & \cdots \\[5pt] 
    d_i^2 & A^1 B^2 & \cdots
  \end{array}
  \]
  The sets $A^j$, $B^j$ and $D^j$ in the preferences described above indicate that all players in these sets appear consecutively in some arbitrary order in the preferences. Ellipses indicate that all remaining preferences may be completed arbitrarily. For example, $a_1^1[1]$ most prefers $w_1^1 x_1^1$, followed by $y_1^1 z_1^1$ and $b_{k_m}^1 d_{\ell_m}^1$. She then prefers all remaining pairs in $B^1 D^1$ in any order, followed by all pairs in $B^2 D^2$, followed by the remaining pairs in any order.

  \begin{lem}
  The embedding $f :$ $3$DM-$3\too 3$GSM described above satisfies
    \begin{enumerate}
    \item If $\opt(I) = m$---that is, $I$ admits a perfect matching---then $f(I)$ admits a stable marriage (i.e. $\MSM(P) = n^3$).
    \item If $\opt(I) \leq c m$ for some $c < 1$, then there exists a constant $c' < 1$ depending only on $c$ such that $\MSM(P) \leq c' n^3$.
    \end{enumerate}
    \label{lem:embedding}
  \end{lem}
  \begin{proof}
    To prove the first claim assume, without loss of generality, that $M' = \set{a_i b_{k_1} d_{\ell_1} \st i \in [n]}$ is a perfect matching in $E$. It is easy to verify the marriage
    \[
    M = \set{a_i^j[1] b_{k_1}^j d_{\ell_1}^j} \cup \set{a_i^j[2] w_i^j x_i^j} \cup \set{a_i^j[3] y_i^j z_i^j}
    \]
    contains no blocking triples, hence is a stable marriage.

    For the second claim, let $M$ be an arbitrary marriage in $A \times B \times D$. We observe that there are at least $(1 - c) m$ players $a^1 \in A^1$ and $(1 - c) m$ players $a^2 \in A^2$ that are not matched with pairs from their top three choices. Suppose to the contrary that $\alpha > (2 + c) m$ players $a^1 \in A^1$ are matched with their top $3$ choices. This implies that more than $c m$ women $a^1 \in A^1$ are matched in triples of the form $a^1 b^1_{k} d^1_{\ell}$ with $a b_k d_\ell \in E$, implying that $E$ contains a matching of size $\alpha - 2 m > c m$, a contradiction. Thus at least $2 (1 -c) m$ women in $A^1 \cup A^2$ are matched below their top three choices.

    Let $A'$ denote the set of women matched below their top three choices, and $B'$ and $D'$ the sets of partners of $a \in A'$ in $M$. By the previous paragraph, $\abs{A'} \geq 2 (1 - c)m = (1 - c)m/6$. Further, the relative preferences of players in $A'$, $B'$ and $D'$ are precisely those described in Theorem \ref{thm:msm}. Thus, by Theorem $\ref{thm:msm}$, any marriage $M$ among these players induces at least $\abs{A'}^3 / 128$ blocking triples. Hence $M$ must contain at least 
    \[
    \frac{\abs{A'}}{128} \geq \frac{(1 - c)^3}{3456} n^3
    \]
    blocking triples.
  \end{proof}

  \begin{proof}[Proof of Theorem \ref{thm:main-3gsm-lb}]
    The reduction $f :$ $3$DM-$3\too 3$GSM is easily seen to be polynomial time computable. Thus, by Lemma \ref{lem:embedding}, $f$ is a polynomial time reduction from \GTDMT{c}\ to \GTGMSM{c'} where $c' = 1 - (1 - c)^3/3456$. The NP hardness of \GTGMSM{c}\ is then an immediate consequence of Theorem \ref{thm:gap-3dm-3-hardness}. 

    The hardness of \GTGMSS{c}\ is a consequence of the hardness \GTGMSM{c}. Consider an instance of \TGSM\ with preferences $P$. We make the following observations.
    \begin{enumerate}
    \item $\MSM(P) = n^3$ if and only if $\MSS(P) = n$.
    \item If $\MSM(P) \leq (1 - 3 \e) n^3$ for $\e > 0$, then $\MSS(P) \leq (1 - \e) n$.
    \end{enumerate}
    The first observation is clear. To prove the second, suppose that $\MSS(P) > (1 - \e) n$, and let $S$ be a maximum stable submarriage. We can form a marriage $M$ by arbitrarily adding $\e n$ disjoint families to $S$. Since each new family can induce at most $3 n^2$ blocking triples, $M$ has at most $3 \e n^3$ blocking triples, hence $\MSM(P) > (1 - 3 \e) n^3$. The two observations above imply that any decider for \GTGMSS{(1 - \e)} is also a decider for \GTGMSM{(1 - 3 \e)}. Thus, the NP-hardness of \GTGMSM{c} immediately implies the analogous result for \GTGMSS{c}.
  \end{proof}

  A sketch of the proof of analogous lower bounds for $3$PSA is given in Appendix \ref{sec:gen-to-3sr}.

  %%%%%%%%%%%%%%%%%%%%%%%%%%%%%%%%%%%%%%%%%%%%%%%%%%%%%%%%%%%%%%%%%%%%%
  %% Approximation Algorithm
  %%%%%%%%%%%%%%%%%%%%%%%%%%%%%%%%%%%%%%%%%%%%%%%%%%%%%%%%%%%%%%%%%%%%%

  \section{Approximation Algorithms}
  \label{sec:approximation}

  \subsection{\TGSM\ approximation}

  In this section, we describe a polynomial time approximation algorithm for MSM, thereby proving Theorem \ref{thm:amsm-informal}. Consider an instance of \TGSM\ with preferences $P$, and as before $A = \set{a_1, a_2, \ldots, a_n}$, $B = \set{b_1, b_2, \ldots, b_n}$, and $D = \set{d_1, d_2, \ldots, d_n}$. Given a triple $a_i b_j d_k$, we define its \dft{stable set} $S_{ijk}$ to be the set of (indices of) triples which cannot form unstable triples with $a_i b_j d_k$. Specifically, we have
  \begin{align*}
    S_{ijk} = &\set{\alpha \beta \delta \in [n]^3 \st b_\beta d_\delta \preceq_{a_i} b_j d_k,\ \alpha = i} \cup \set{\alpha \beta \delta \in [n]^3 \st a_\alpha d_\delta \preceq_{b_j} a_i d_k,\ \beta = j}\\
    &\cup \set{\alpha \beta \delta \in [n]^3 \st a_\alpha b_\beta \preceq_{d_k} a_i b_j,\ \delta = k}
  \end{align*}
  The idea of our algorithm is to greedily form families that maximize $\abs{S_{ijk}}$. Pseudocode is given in Algorithm \ref{alg:amsm}.

  \begin{algorithm}
    \caption{\AMSM$(A, B, D, P)$}
    \label{alg:amsm}
    \begin{algorithmic}
      \STATE find $ijk \in [n]^3$ which maximize $\abs{S_{ijk}}$
      \STATE $A' \leftarrow A \setminus \set{a_i}$, $B' \leftarrow B \setminus \set{b_j}$, $D' \leftarrow D \setminus \set{d_k}$
      \STATE $P' \leftarrow P$ restricted to $A'$, $B'$, and $D'$
      \RETURN $\set{a_i b_j d_k} \cup$\AMSM$(A', B', D', P')$
    \end{algorithmic}
  \end{algorithm}

  It is easy to see that \AMSM\ can be implemented in polynomial time. The naive algorithm for computing $\abs{S_{ijk}}$ for fixed $ijk \in [n]^3$ by iterating through all triples $\alpha\beta\delta \in [n]^3$ and querying each player's preferences can be implemented in time $\sOh(n^3)$. The maximal such $\abs{S_{ijk}}$ can then be found by iterating through all $ijk \in [n]^3$. Thus the first step in \AMSM\ can be accomplished in time $\sOh(n^6)$. Finally, the recursive step of \AMSM\ terminates after $n$ iterations, as each iteration decreases the size of $A$, $B$, and $D$ by one. 

  \begin{lem}
    \label{lem:good-edges}
    For any preferences $P$, and sets $A$, $B$ and $D$ with $\abs{A} = \abs{B} = \abs{D} = n$, there exists a triple $ijk \in [n]^3$ with 
    \begin{equation}
      \label{eqn:good-edges}
      \abs{S_{ijk}} \geq \frac{4n^2}{3} - n - 1.      
    \end{equation}
  \end{lem}
  \begin{proof}
    We will show that there exists a triple $a_i b_j d_k$ such that at least two of $a_i$, $b_j$, and $d_k$ respectively rank $b_j d_k$, $a_i d_k$, and $a_i b_j$ among their top $n^2 / 3 + 1$ choices. Note that this occurs precisely when at least two of the the following inequalities are satisfied
    \begin{align*}
    & \abs{\set{\beta\delta \in [n]^2 \st b_\beta d_\delta \succeq_{a_i} b_j d_k}} \leq \frac{n^2}{3} + 1, \quad \abs{\set{\alpha\delta \in [n]^2 \st a_\alpha d_\delta \succeq_{b_j} a_i d_k}} \leq \frac{n^2}{3} + 1,\\
    &\text{and}\quad \abs{\set{\alpha\beta \in [n]^2 \st a_\alpha b_\beta \succeq_{d_k} a_i b_j}} \leq \frac{n^2}{3} + 1
    \end{align*}
    Mark each triple $a_i b_j d_k$ which satisfies one of the above inequalities. Each $a_i$ induces $\frac{n^2}{3} + 1$ marks, so we get $\frac{n^3}{3} + n$ marks from all $a \in A$. Similarly, we get $\frac{n^3}{3} + n$ marks from $B$ and $D$. Thus, marks are placed on at least $n^3 + 3n$ triples. By the pigeonhole principle, at least one triple is marked twice.

    We claim that the triple $a_i b_j d_k$ satisfying two of the above inequalities satisfies equation~(\ref{eqn:good-edges}). Without loss of generality, assume that $a_i b_j d_k$ satisfies the first two equations. Thus, $a_i$ and $b_j$ must each contribute at least $\frac{2 n^3}{3} - 1$ stable triples with respect to $a_i b_j d_k$. Further, at most $n - 1$ such triples can be contributed by both $a_i$ and $b_j$, as such triples must be of the form $a_i b_j d_\delta$ for $\delta \neq k$. Thus (\ref{eqn:good-edges}) is satisfied, as desired.
  \end{proof}

  We are now ready to prove that \AMSM\ gives a constant factor approximation for the maximally stable marriage problem.

  \begin{proof}[Proof of Theorem \ref{thm:amsm-informal}]
    Let $M$ be the marriage found by \AMSM, and suppose
    \[
    M = \set{a_1 b_1 d_1,\ a_2 b_2 d_2,\ \ldots,\ a_n b_n d_n}
    \]
    where $a_1 b_1 d_1$ is the first triple found by \AMSM, $a_2 b_2 d_2$ is the second, et cetera. 
    By Lemma~\ref{lem:good-edges}, $\abs{S_{111}} \geq \frac{4}{3}n^2 - O(n)$. Therefore, the players $a_1$, $b_1$, and $d_1$ can be contained in at most $\frac{5}{3} n^2 + O(n)$ unstable triples in any marriage containing the family $a_1 b_1 d_1$. Similarly, for $1 \leq i \leq n$ the $i$th family $a_i b_i d_i$ can contribute at most $\frac{5}{3}(n - i + 1)^2 + O(n)$ new unstable triples (not containing any $a_j$, $b_j$, or $d_j$ for $j < i$). Thus, the total number of unstable triples in $M$ is at most
    \[
    \sum_{i = 1}^n \paren{\frac{5}{3}(n - i + 1)^2 + O(n)} = \frac{5}{9} n^3 + O(n^2).
    \]
    Thus, we have $\stab(M) \geq 4 n^3 / 9 - O(n^2)$ as desired.
  \end{proof}
  
  A proof of the analogous result for $3$PSA (Theorem~\ref{thm:asr-informal}) is given in Appendix \ref{sec:3psa-approx}.

  %%%%%%%%%%%%%%%%%%%%%%%%%%%%%%%%%%%%%%%%%%%%%%%%%%%%%%%%%%%%%%%%%%%%%
  %% Conclusion
  %%%%%%%%%%%%%%%%%%%%%%%%%%%%%%%%%%%%%%%%%%%%%%%%%%%%%%%%%%%%%%%%%%%%%

  \section{Concluding Remarks and Open Questions}
  \label{sec:conclusion}

    While \AMSM\ gives a simple approximation algorithm for \TGMSM, we do not generalize this result to \TGMSS. Indeed, even the first two families output by \AMSM\ may include blocking triples. We leave the existence of an efficient approximation for \TGMSS\ as a tantalizing open question.

  \begin{open}
    Is it possible to efficiently compute a constant factor approximation to \TGMSS?
  \end{open}
  
  Finding an approximation algorithm for maximally stable marriage was made easier by the fact that any preferences admit a marriage/matching with $\Omega(n^3)$ stable triples. However, for \TGMSS, it is not clear whether every preference structure admits stable submarriages of size $\Omega(n)$. We feel that understanding the approximability of \TGMSS\ is a very intriguing avenue of further exploration.

  \begin{open}
    How small can a maximum stable submarriage/submatching be? What preferences achieve this bound?
  \end{open}

  In our hardness of approximation results (Theorems \ref{thm:main-3gsm-lb} and \ref{thm:main-3sr-lb}), we do not state explicit values of $c$ for which \GTGMSM{c}\ and \GTGMSS{c}\ (and the corresponding problems for three person stable assignment) are NP-complete. The value implied by our embedding of $3$SAT-$B$ via $3$DM-$3$ is quite close to $1$. It would be interesting to find a better (explicit) factor for hardness of approximation. Conversely, is it possible to efficiently achieve a better than $4/9$-factor approximation for maximally stable marriage/matching problems?

  \begin{open}
    For the maximally stable marriage/matching problems, close the gap between the $4/9$-factor approximation algorithm and the $(1-\e)$-factor hardness of approximation.
  \end{open}

  The preference structure described in the proof of Theorem \ref{thm:msm} (upon which our hardness of approximation results rely), there exist many quartets of pairs $b_1 d_1$, $b_2 d_2$, $b_1, d_3$, and $b_2 d_4$ such that $b_1 d_1 \succ_a b_2 d_2$ and $b_2 d_4 \succ_a b_1 d_3$. Thus $a$ does not consistently prefer pairs including $b_1$ to those including $b_2$ or vice versa. Ng and Hirschberg describe such preferences as \dft{inconsistent}, and asked whether consistent preferences always admit a ($3$ gender) stable marriage. Huang \cite{Huang07} showed that consistent preferences need not admit stable marriages, and indeed it is NP-complete to determine whether or not given consistent preferences admit a stable marriage. 

  \begin{open}
    Are MSM and MSS still hard to approximate if preferences are restricted to be consistent? 
  \end{open}

  \bibliographystyle{plain}
  \bibliography{3GSM}

  %%%%%%%%%%%%%%%%%%%%%%%%%%%%%%%%%%%%%%%%%%%%%%%%%%%%%%%%%%%%%%%%%%%%%
  %% Appendix
  %%%%%%%%%%%%%%%%%%%%%%%%%%%%%%%%%%%%%%%%%%%%%%%%%%%%%%%%%%%%%%%%%%%%%

  \appendix

  %%%%%%%%%%%%%%%%%%%%%%%%%%%%%%%%%%%%%%%%%%%%%%%%%%%%%%%%%%%%%%%%%%%%%
  %% Hardness of Gap-3DM-3
  %%%%%%%%%%%%%%%%%%%%%%%%%%%%%%%%%%%%%%%%%%%%%%%%%%%%%%%%%%%%%%%%%%%%%

  \section{Hardness of \GTDMT{c}}
  \label{sec:gap-hardness}

  The goal of this section is to prove Theorem \ref{thm:gap-3dm-3-hardness}, the NP hardness of \GTDMT{c}. In the first subsection, we review Kann's reduction from Max-$3$SAT-$B$ to Max-$3DM$-$3$. In the following subsection, we describe how to modify Kann's reduction in order to obtain Theorem \ref{thm:gap-3dm-3-hardness}. 

  \subsection{Kann's reduction} Let $I$ be an instance of $3$SAT-$B$. Specifically, $I$ consists of a set $U$ of $n$ Boolean variables,
  \[
  U = \set{u_1, u_2, \ldots, u_n},
  \]
  and a set $C$ of $m$ clauses
  \[
  C = \set{c_1, c_2, \ldots, c_m}.
  \]
  Each clause is a disjunction of at most $3$ literals, and each variable $u_i$ or its negation $\bar{u}_i$ appears in at most $B$ clauses. For each $i \in [n]$, let $d_i$ denote the number of clauses in which $u_i$ or $\bar{u}_i$ appears, where $d_i \leq B$. Kann's construction begins with the classical reduction from $3$SAT to $3$DM used to show the NP-completeness of $3$DM, as described in \cite{GJ79}. Each variable $u_i$ gets mapped to a \dft{ring} of $2 d_i$ edges. The points of the ring correspond alternatingly to $u_i$ and $\bar{u}_i$. A maximal matching on the ring corresponds to a choice a truth value of $u_i$: if the edges containing the vertices labeled $u_i$ are in the matching, this corresponds to $u_i$ having the value true; if the edges containing $\bar{u}_i$ are chosen, this corresponds to $u_i$ taking the value false. See Fig \ref{fig:ring}. 

  In the classical construction, the points of the ring corresponding to $u_i$ are connected to \dft{clause vertices} via \dft{clause edges} which encode the clauses in $C$. This is done in such a way that the formula $I$ is satisfiable if and only if the corresponding matching problem admits a perfect matching. The problem with this embedding, however, is that even if a relatively small fraction of clauses in $C$ can be simultaneously satisfied, the corresponding matching problem may still admit a nearly perfect matching. 

  To remedy this problem, Kann's reduction maps each Boolean variable $u_i$ to many rings. The rings are then connected via a tree structure whose roots correspond to instances of $u_i$ in the clauses of $C$. This tree structure imposes a predictable structure on the maximal matchings. 

  We denote the parameter
  \[
  K = 2^{\floor{\log(3 B / 2 + 1)}}
  \]
  which is the number of rings to which each variable $u_i$ maps. We denote the ``free elements''---the points on the rings---associate to $u_i$ by the variables
  \[
  v_i[\gamma, k] \text{ and } \bar{v}_i[\gamma, k] \quad\text{for}\quad 1 \leq \gamma \leq d_i,\ 1 \leq k \leq K.
  \]
  These vertices are connected to rings as in Figure \ref{fig:ring}. The rings are connected via \dft{tree edges} in $2 d_i$ binary trees, such that for each fixed $\gamma$, $v_i[\gamma, 1], v_i[\gamma, 2], \ldots, v_i[\gamma, K]$ are the leaves of a tree, and similarly for the $\bar{v_i}[\gamma, k]$. We label the root of this tree by $u_i[\gamma]$ or $\bar{u}_i[\gamma]$ depending on the labels of the leaves. See Figure \ref{fig:ring-of-trees}. We refer to the resulting structure for $u_i$ as the \dft{ring of trees} corresponding to $u_i$. 

	\begin{figure}[t]
          \begin{center}
            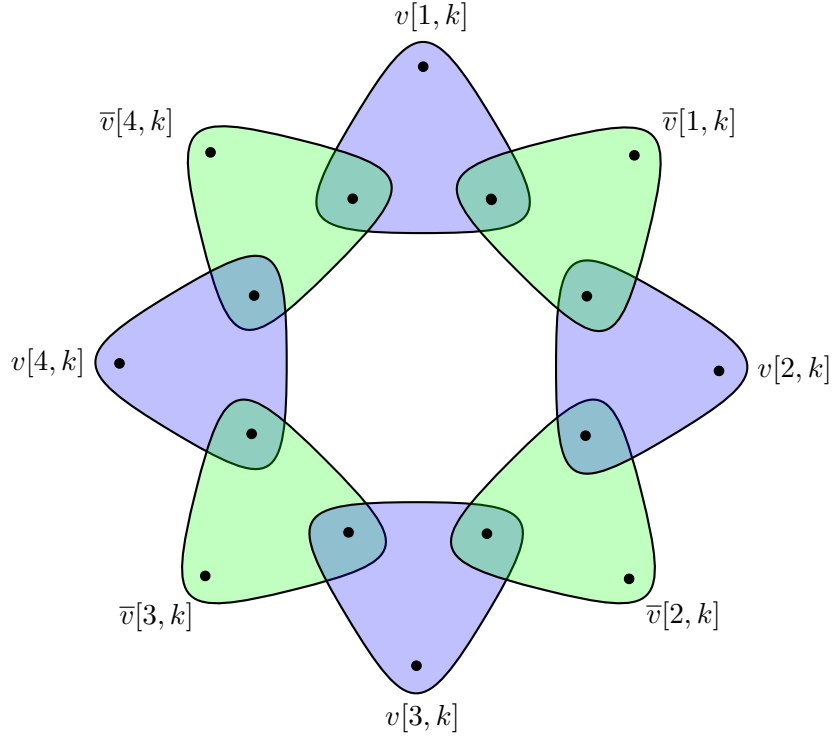    
          \end{center}
	  \caption{\label{fig:ring} The \dft{ring} structure for the embedding of $3$SAT-$B$ into $3$DM-$3$. The ring shown corresponds to a variable $u_i$ with $d_i = 4$. An optimal matching in the ring corresponds to a truth value of for the variable $u_i$: the blue edges correspond to the value \emph{true} while the green edges correspond to the value \emph{false}.}
	\end{figure}
	
	\begin{figure}[t]
          \begin{center}
	    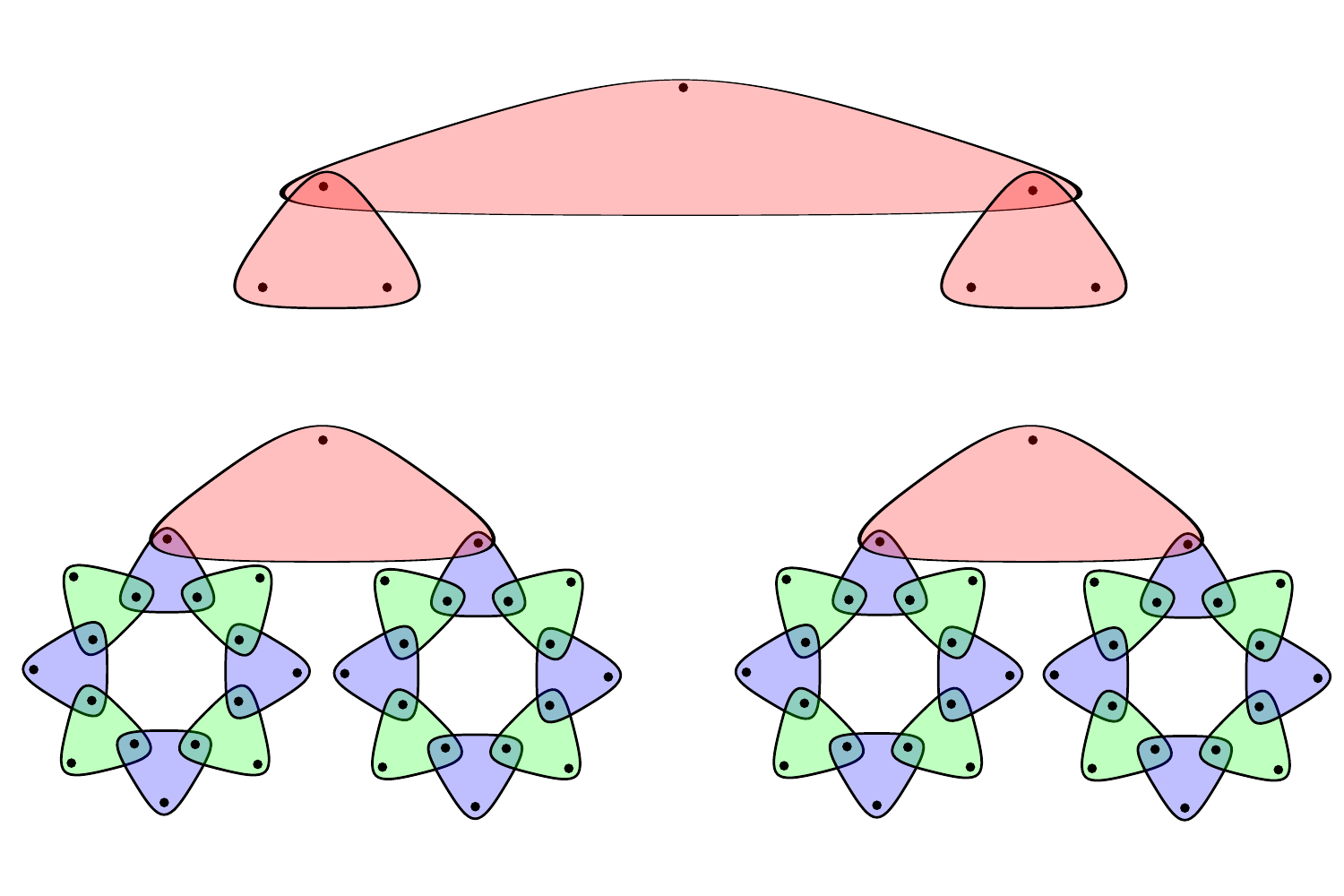            
          \end{center}
	  \caption{\label{fig:ring-of-trees} The \dft{ring of trees} structure for Kann's embedding of $3$SAT-$B$ into $3$DM-$3$. The red edges are \dft{tree edges}. In addition to the tree shown, the ring of trees contains (identical) trees for each $u_i[\gamma]$ and $\bar{u}_i[\gamma]$. In the example pictured, $\gamma$ ranges from $1$ to $4$. The \dft{root vertices} are labeled $u_i[1], \bar{u}_i[1], \ldots, u_i[d_i], \bar{u}_i[d_i]$ where $d_i$ is the number of occurrences of $u_i$ or its negation in $I$. As described in Lemma \ref{lem:opt-matchings}, an optimal matching in the ring of trees can be obtained by a consistent choice of all blue or green vertices in all the rings associated to $u_i$, then covering as many remaining vertices as possible with tree edges (a greedy ``leaf to root'' approach works). It is clear that such a matching will cover all vertices, except for half of the root vertices (those corresponding to either $u_i$ or $\bar{u}_i$).}
	\end{figure}

  The root vertices are connected via \dft{clause edges} to \dft{clause vertices}. For each $c_j \in C$, we associate two clause vertices $s_1[j]$ and $s_2[j]$. If $c_j$ is $u_i$'s $\gamma$-th clause in $C$, then we include the edge
  \[
  \set{u_i[\gamma], s_1[j], s_2[j]} \quad\text{or}\quad \set{\bar{u}_i[\gamma], s_1[j], s_2[j]}
  \]
  depending if $u_i$ or its negation appears in $c_j$ and the parity of the of the tree of rings. We denote the resulting instance of $3$DM by $f(I)$. It is readily apparent from this construction that $f(I)$ is in fact an instance of $3$DM-$3$: all vertices in the rings of trees are contained in exactly $2$ edges, while clause vertices are contained in at most $3$ edges. Further, the vertex set $V$ can be partitioned into a disjoint union $W \cup X \cup Y$ such that each edge contains one vertex from each of these sets. Kann makes the following observations about the structure of optimal (maximum) matchings in $f(I)$.

  \begin{lem}
    \label{lem:opt-matchings}
    Let $I$ be an instance of $3$SAT-$B$. Let $f(I)$ be an instance of $3$DM-$3$ constructed as above. Then each optimal matching $M$ in $f(I)$ is associated with an optimal assignment in $I$, and has the following structure.
    \begin{enumerate}
    \item For each variable $u_i$, $M$ contains either all edges containing $v_i[\gamma, k]$ or all the edges containing $\bar{v}_i[\gamma,  k]$, depending on the value $u_i$ in the optimal assignment for $I$.
    \item From each ring of trees, alternating tree edges are included in $M$ so as to cover all tree (and ring) vertices, except possibly root vertices.
    \item If $c_j$ is satisfied in the optimal assignment in $I$, then $M$ contains an edge containing $s_1[j]$ and $s_2[j]$.
    \item If $c_j$ is unsatisfied in the optimal assignment in $I$, then none of the edges containing $s_1[j]$ and $s_2[j]$ are contained in $M$.
    \end{enumerate}
    In particular, the only possible vertices left uncovered in an optimal matching are clause vertices corresponding to unsatisfied clauses and root vertices.
  \end{lem}

  As a consequence of Kann's analysis of the optimal matchings in $f(I)$, he is able to show that $f$ is an $L$-reduction from $3$SAT-$B$ to $3$DM-$3$.

  \subsection{Modification of Kann's reduction} 

  In this section, we describe a reduction $f' : 3$SAT-$B \too 3$DM-$3$ such that for each satisfiable instance $I$ of $3$SAT-$B$, $f'(I)$ admits a perfect matching. In the reduction $f$ above, even if $I$ is satisfiable, there may be many root vertices that are not in an optimal matching, $M$. In particular, if a clause $c_j$ is satisfied by $u_{i}$ and $u_{i'}$, then at most one of the edges
  \[
  \set{u_{i}[\gamma],\ s_1[j],\ s_2[j]} \quad\text{and}\quad\set{u_{i'}[\gamma'],\ s_1[j],\ s_2[j]}
  \]
  can appear in $M$. Hence, at most one of $u_i[\gamma]$ and $u_{i'}[\gamma']$ can appear in $M$. To remedy this problem, we define $f'(I)$ to be the disjoint union of three copies of $f(I)$
  \[
  f'(I) = f(I)_1 \sqcup f(I)_2 \sqcup f(I)_3.
  \]
  We then add an edge for each root vertex in $f(I)$ that contains the corresponding root vertices in each disjoint copy of $f(I)$. Specifically, if $u_{i}[\gamma]_1$, $u_i[\gamma]_2$, and $u_{i}[\gamma]_3$ are the three copies in $f'(I)$ of a root vertex $u_i[\gamma]$ in $f(I)$, then we include the edge
  \begin{equation}
    \label{eqn:root-edge}
    \set{u_{i}[\gamma]_1,\ u_{i}[\gamma]_2,\ u_{i}[\gamma]_3}    
  \end{equation}
  in $f'(I)$. We now describe the structure of optimal matchings $M'$ in $f'(I)$. Let $M_1$, $M_2$, and $M_3$ be the restrictions of an optimal matching $M'$ for $f'(I)$ to $f(I)_1$, $f(I)_2$, and $f(I)_3$ respectively. Thus, we can write
  \begin{equation}
    \label{eqn:new-matching}
    M' = M_1 \cup M_2 \cup M_3 \cup R    
  \end{equation}
  where $R$ contains those edges in $M'$ of the form (\ref{eqn:root-edge}). 

  \begin{lem}
    \label{lem:symmetry}
    There exists an optimal matching $M'$ for $f'(I)$ such that the matchings $M_1$, $M_2$, and $M_3$ contain precisely the same edges as an optimal matching $M$ for $f(I)$.
  \end{lem}
  \begin{proof}
    Suppose $M'$ is an optimal matching for $f'(I)$. We may assume without loss of generality that the matchings $M_1$, $M_2$ and $M_3$ are all identical to some matching $M$ on $f(I)$. Indeed, if, say $M_1$ is the largest of the three matchings, we can increase the size of $M'$ by replacing $M_2$ and $M_3$ with identical copies of $M_1$. Since the only edges between $M_1$, $M_2$, and $M_3$ are edges of the form (\ref{eqn:root-edge}), replacing $M_2$ and $M_3$ with copies of $M_1$ cannot decrease the size of $M'$. Thus, we may assume that
    \[
    \abs{M'} = 3 \abs{M} + \abs{R}
    \]
    where $M$ is some matching on $f(I)$, and $R$ consists of edges of the form (\ref{eqn:root-edge}).

    We will now argue that $M$ is indeed an optimal matching on $f(I)$, hence has the form described in Lemma \ref{lem:opt-matchings}. Notice that if $M$ is optimal for $f(I)$, then by including all edges in $R$ containing uncovered root vertices, $M'$ covers every ring, tree, and root vertex. Thus, the only way to obtain a larger matching would be to include more clause edges. However, by Lemma \ref{lem:opt-matchings}, including more clause edges cannot increase the size of the matching $M$. Thus, we may assume $M$ is an optimal matching for $f(I)$.
  \end{proof}

  \begin{cor}
    \label{cor:optimal-matching}
    If $I$ is an instance of $3$-SAT-$B$ with $m$ clauses and $\opt(I) = c m$ for some $c \leq 1$, then an optimal matching $M'$ in $f'(I)$ leaves precisely $6 (1 - c) m$ vertices uncovered.
  \end{cor}

  \begin{lem}
    \label{lem:size-bound}
    There exists a constant $C > 0$ depending only on $B$ such that the number of vertices in $f'(I)$ is at most $C m$.
  \end{lem}
  \begin{proof}
    We bound the number of ring, tree, and clause vertices separately. Since the vertex set of $f'(I)$ consists of three disjoint copies of the vertices in $f(I)$, it suffices to bound the number of vertices in $f(I)$.
    \begin{description}
    \item[Ring vertices] For each variable $u_i$, there are $K = O(B)$ rings, each consisting of $4 d_i = O(B)$. Thus there are $O(B^2)$ ring vertices for each variable $u_i$, hence a total of $O(B^2 n)$ ring vertices in $f(I)$.
    \item[Tree vertices] Since each ring vertex of the form $v_i[\gamma, k]$ is the leaf of a complete binary tree whose internal nodes and root are tree vertices, there are $O(B^2 n)$ tree vertices in $f(I)$.
    \item[Clause vertices] There two vertices $s_1[j]$ and $s_2[j]$ associated to each of $m$ clauses, hence there are $O(m)$ clauses in total.
    \end{description}
    Therefore, the total number of vertices in $f(I)$ and hence $f'(I)$ is $O(m + B^2 n)$. Clearly, we may assume that $n \leq m$, so that there are $O(B^2 m)$ vertices in $f'(I)$.
  \end{proof}

  \begin{cor}
    \label{cor:gap}
    Let $I$ be an instance of $3$SAT-$B$, and let $M^* = \frac 1 3 \abs{f'(I)} = \abs{f(I)}$ be the number of vertices in $f(I)$. Then for any $c < 1$, there exists a constant $c' < 1$ depending only on $c$ and $B$ such that:
    \begin{itemize}
    \item if $\opt(I) = m$ (i.e., $I$ is satisfiable) then $\opt(f'(I)) = M^*$;
    \item if $\opt(I) \leq c m$ then $\opt(f'(I)) \leq c' M^*$.
    \end{itemize}
  \end{cor}

  Theorem \ref{thm:gap-3dm-3-hardness} follows immediately from Corollary \ref{cor:gap} and the following incarnation of the PCP theorem.

  \begin{thm}[PCP Theorem \cite{AS98, ALMSS98, Dinur07}]
    There exist a absolute constants $c < 1$ and $B$ such that it is NP-hard to determine if an instance $I$ of $3$SAT-$B$ satisfies $\opt(I) = m$ or $\opt(I) \leq c m$.
  \end{thm}

  %%%%%%%%%%%%%%%%%%%%%%%%%%%%%%%%%%%%%%%%%%%%%%%%%%%%%%%%%%%%%%%%%%%%%
  %% Preferences with Many Unstable Triples
  %%%%%%%%%%%%%%%%%%%%%%%%%%%%%%%%%%%%%%%%%%%%%%%%%%%%%%%%%%%%%%%%%%%%%

  \section{Preferences with Many Unstable Triples}
  \label{sec:unstable}

  We first consider the case where $n = 2$. We denote $A = \set{a_1, a_2}$, $B = \set{b_1, b_2}$, and $D = \set{d_1, d_2}$. Consider preference lists $P$ as described in the following table, where most preferred partners are listed first.
  \begin{center}
    \begin{tabular}{l|rrr}
      player & \multicolumn{3}{c}{preferences}\\
      \hline
      $a_1$ & $b_1 d_1$ & $b_2 d_2$ & $\cdots$\\
      $a_2$ & $b_2 d_1$ & $\cdots$\\
      $b_1$ & $a_1 d_1$ & $\cdots$\\
      $b_2$ & $a_1 d_2$ & $a_2 d_1$ & $\cdots$\\
      $d_1$ & $a_2 b_2$ & $a_1 b_1$ & $\cdots$\\
      $d_2$ & $a_1 b_2$ & $\cdots$
    \end{tabular}      
  \end{center}
  The ellipses indicate that the remaining preferences are otherwise arbitrary. Suppose $M$ is a stable marriage for $P$. We must have either $a_1 b_1 d_1 \in M$ or $a_1 b_2 d_2 \in M$, for otherwise the triple $a_1 b_2 d_2$ is unstable. However, if $a_1 b_1 d_1 \in M$, then $a_2 b_2 d_1$ is unstable. On the other hand, if $a_1 b_2 d_2 \in M$ then $a_1 b_1 d_1$ is unstable. Therefore, no such stable $M$ exists. In particular, every marriage $M$ contains at least one unstable triple.

  The idea of the proof of Theorem \ref{thm:msm} is to choose preferences $P$ such that when restricted to many sets of two women, two men and two dogs, the preferences are as above. Thus any marriage containing families consisting of these players must induce unstable triples.

  \begin{proof}[Proof of Theorem \ref{thm:msm}]
    We partition the sets $A$, $B$ and $D$ each into two sets of equal size: $A = A_1 \cup A_2$, $B = B_1 \cup B_2$, $D = D_1 \cup D_2$. Consider the preferences $P$ described in Section~\ref{sec:many-blockers}. We will prove that for $P$, every matching $M$ contains at least $n^3 / 128$ unstable triples. Let $M$ be an arbitrary marriage, and suppose $\ins(M) < n^3/128$. We consider two cases separately.
    \begin{description}
    \item[Case 1: $\abs{M \cap (A_1 \times B_1 \times D_1)} \leq n / 4$.] Let $A_1'$, $B_1'$ and $D_1'$ be the subsets of $A_1$, $B_1$ and $D_1$ respectively of players not in triples contained in $A_1 \times B_1 \times D_1$. By the hypothesis, $\abs{A_1'}, \abs{B_1'}, \abs{D_1'} \geq n / 4$. Let $d_1 \in D_1'$. Notice that if $p_M(d_1) \notin A_2 \times B_2$, then $a_1 b_1 d_1$ is unstable for all $a_1 \in A_1'$, $b_2 \in B_1'$. Since fewer than $n^3 / 128$ triples in $A_1' \times B_1' \times D_1'$ are unstable, at least $3 n / 8$ dogs $d_1 \in D_1'$ must have families $a_2 b_2 d_1 \in A_2 \times B_2 \times D_1'$.

      Since $\abs{M \cap (A_2 \times B_2 \times D_1)} \geq 3 n / 8$, we must have $\abs{M \cap (A_1 \times B_2 \times D_2)} \leq n / 8$. Thus, there must be at least $n / 8$ women $a_1 \in A_1$ with partners not in $(B_1 \times D_1) \cup (B_2 \times D_2)$. However, every such $a_1$ forms an unstable triple with every $b_2 \in B_2$ and $d_2 \in D_2$ which are not in families in $A_1 \times B_2 \times D_2$. Since there at least $3 n / 8$ such $b_2$ and $d_2$, there are at least
      \[
      \paren{\frac n 8}\paren{\frac{3 n} 8}\paren{\frac{3n} 8} > \frac n {128}
      \]
      blocking triples, a contradiction.

    \item[Case 2: $\abs{M \cap (A_1 \times B_1 \times D_1)} > n  / 4$.] In this case, we must have $\abs{M \cap (A_2 \times B_2 \times D_1)} < n / 4$. This implies that 
      \begin{equation}
        \label{eqn:marriage-size}
        \abs{M \cap (A_1 \times B_2 \times D_2)} > 3 n / 8  
      \end{equation}
      for otherwise triples $a_2 b_2 d_1 \in (A_2 \times B_2 \times D_1)$ with $p_M(b_2) \notin A_1 \times D_1$ form more than $n^3 / 128$ unstable triples. But (\ref{eqn:marriage-size}) contradicts the Case 2 hypothesis, as $\abs{A_1} = n / 2$. 
    \end{description}
    Since both cases lead to a contradiction, we may conclude that any $M$ contains at least $n^3 / 128$ unstable triples, as desired.
  \end{proof}

  %%%%%%%%%%%%%%%%%%%%%%%%%%%%%%%%%%%%%%%%%%%%%%%%%%%%%%%%%%%%%%%%%%%%%
  %% 3 PSA Results
  %%%%%%%%%%%%%%%%%%%%%%%%%%%%%%%%%%%%%%%%%%%%%%%%%%%%%%%%%%%%%%%%%%%%%

  \section{Results for $3$PSA}
  \label{sec:3psa}

  In this appendix, we prove our main results for $3$PSA, Theorems \ref{thm:main-3sr-lb} and \ref{thm:asr-informal}. 

  \subsection{\TPSA\ hardness of approximation}
  \label{sec:gen-to-3sr}

  \begin{proof}[Proof sketch of Theorem \ref{thm:main-3sr-lb}]
    As noted in Remark \ref{rem:3sr}, we may view \TGSM\ as a special case of \TPSA\ with incomplete preferences. The NP-hardness of approximation of \TPSA\ with incomplete preferences is analogous to the proof of Theorem \ref{thm:main-3gsm-lb}. Given an instance $I$ of \TGSM\ with sets $A$, $B$, and $D$ and preferences $P$, take $U = A \cup B \cup D$ and form \TPSA\ preferences $P'$ by appending the remaining pairs to $P$ arbitrarily. Analogues of Theorem \ref{thm:msm} and Lemma \ref{lem:embedding} hold for this instance of \TPSA, whence Theorem \ref{thm:main-3sr-lb} follows. We leave details to the reader.    
  \end{proof}

  \subsection{\TPSA\ approximation}
  \label{sec:3psa-approx}

  \AMSM\ can easily be adapted for \TPSA. Let $U$ be a set of players with $\abs{U} = 3n$, and let $P$ be a set of complete preferences for the players in $U$. Given a triple $abc \in \binom{U}{3}$, we form the \dft{stable set} $S_{abc}$ consisting of triples that at least one of $a$, $b$, $c$ does not prefer to $abc$. The approximation algorithm \ASR\ for \TPSA\ is analogous to \AMSM: form a matching $M$ by finding a triple $abc$ that maximizes $\abs{S_{abc}}$, then recursing. The following lemma and its proof are analogous to Lemma~\ref{lem:good-edges}.

  \begin{lem}
    \label{lem:good-triples}
    For any set $U$ of players with $\abs{U} = 3n$ and complete preferences $P$, there exists a triples $abc \in \binom{U}{3}$ such that
    \[
    \abs{S_{abc}} \geq 6 n^2 - O(n).
    \]
  \end{lem}

  Using Lemma~\ref{lem:good-triples}, we prove Theorem~\ref{thm:asr-informal} analogously to Theorem~\ref{thm:amsm-informal}.
  \begin{proof}[Proof of Theorem~\ref{thm:asr-informal}]
    Each triple $abc$ can intersect at most $3 \binom{3n}{2} \leq \frac{27}{2} n^2$ blocking triples. Thus, by Lemma \ref{lem:good-triples}, the total number blocking triples in the matching $M$ found by \ASR\ is at most
    \[
    \sum_{i = 0}^{n-1} \paren{\frac{15}{2} (n - i + 1)^2 + O(n)} = \frac{5}{2} n^3 + O(n^2).
    \]
    Therefore,
    \[
    \stab(M) \geq 2 n^3 - O(n^2),
    \]
    as the total number of triples in $\binom{U}{3}$ is $\frac{9}{2}n^3 - O(n^2)$. Hence $M$ is a $\frac{4}{9}$-approximation to a maximally stable matching, as desired.
  \end{proof}

\end{document}